\documentclass{article}
\usepackage[margin=1in]{geometry}
\usepackage{amsmath,amsfonts,amssymb,dsfont}
\usepackage{amsthm}
\usepackage{url,hyperref}

\usepackage{algorithm,algorithmic}

\newtheorem{theorem}{Theorem}[section]
\newtheorem{definition}{Definition}[section]

\newtheorem{lemma}{Lemma}[theorem]

\def\tr{{\rm tr}} 
\def\rank{{\rm rank}} 
\def\Pr{{\mathbb{P}}} 

\def\R{{\mathbb{R}}} 

\newcommand{\alg}{\mathcal{M}}
\newcommand{\mech}{\mathcal{M}}
\newcommand{\eps}{\varepsilon}
\newcommand{\E}{\mathbb{E}}

\newcommand{\univ}{U}

\def\b0{{\bf 0}}

\newcommand\junk[1]{}

\newcommand{\maps}{\colon}

\DeclareMathOperator{\specLB}{SpecLB}

\DeclareMathOperator{\conv}{conv}
\DeclareMathOperator{\err}{err}
\DeclareMathOperator{\opt}{opt}

\DeclareMathOperator{\polylog}{polylog}

\newcommand{\menote}[1]{}

\newcommand{\tra}{\intercal}

\renewcommand{\univ}{\mathcal{U}}
\newcommand{\quer}{\mathcal{Q}}
\newcommand{\row}{e}

\begin{document}
\title{An Improved Private Mechanism for Small Databases}
\author{Aleksandar Nikolov\\{Microsoft Research}\\{ Redmond, WA,
    USA}\\{alenik@microsoft.com}} 
\date{}
\maketitle              

\begin{abstract}
  We study the problem of answering a workload of linear queries
  $\quer$, on a database of size at most $n = o(|\quer|)$ drawn from a
  universe $\univ$ under the constraint of (approximate) differential
  privacy. Nikolov, Talwar, and Zhang~\cite{NTZ} proposed an efficient
  mechanism that, for any given $\quer$ and $n$, answers the queries
  with average error that is at most a factor polynomial in $\log
  |\quer|$ and $\log |\univ|$ worse than the best possible. Here we
  improve on this guarantee and give a mechanism whose competitiveness
  ratio is at most polynomial in $\log n$ and $\log |\univ|$, and has
  no dependence on $|\quer|$. Our mechanism is based on the projection
  mechanism of~\cite{NTZ}, but in place of an ad-hoc noise
  distribution, we use a distribution which is in a sense optimal for
  the projection mechanism, and analyze it using convex duality and
  the restricted invertibility principle.
\end{abstract}

\section{Introduction}

The central problem of private data analysis is to characterize to
what extent it is possible to compute useful information from
statistical data without compromising the privacy of the individuals
represented in the dataset. In order to formulate this problem
precisely, we need a database model and a definition of what it means
to preserve privacy. Following prior work, we model a database as a
multiset $D$ of $n$ elements from a universe $\univ$\junk{, i.e.~$D
  \in \univ^n$}, with each database element specifying the data of a
single individual. Defining privacy is more subtle. A definition which
has received considerable attention in recent years is
\emph{differential privacy}, which postulates that a randomized
algorithm preserves privacy if its distribution on outputs is almost
the same (in an appropriate metric) on any two input databases $D$ and
$D'$ that differ in the data of at most a single individual. The
formal definition is as follows:
\begin{definition}[\cite{DMNS}]\label{def:DP}
  Two databases $D$ and $D'$ are \emph{neighboring} if the size of
  their symmetric difference is at most one. A randomized algorithm
  $\alg$ satisfies \emph{$(\eps, \delta)$-differential privacy} if for any
  two neighboring databases $D$ and $D'$ and any measurable event $S$
  in the range of $\alg$,
  \begin{equation*}
    \Pr[\alg(D) \in S] \leq e^{\eps}\Pr[\alg(D') \in S] + \delta.
  \end{equation*}
  \junk{Above, probabilities are taken over the internal coin tosses of $\alg$.}
\end{definition}
Differential privacy has a number of desirable properties: it is
invariant under post-processing, the privacy loss degrades smoothly
under (possibly adaptive) composition, and the privacy guarantees hold
in the face of arbitrary side information. We will adopt it as our
definition of choice in this paper. We will work in the regime $\delta
> 0$, which is often called approximate differential privacy, to
distinguish it from pure differential privacy, which is the case
$\delta=0$. Approximate differential privacy provides strong semantic
guarantees when $\delta$ is  $n^{-\omega(1)}$: roughly speaking,
it implies that with probability at least $1 - O(n\sqrt{\delta})$, an
arbitrarily informed adversary cannot guess from the output of the
algorithm if any particular user is represented in the
database. See~\cite{GantaKS08} for a precise formulation of this
semantic guarantee.

We then turn to the question of understanding the constraints imposed
by privacy on the kinds of computation we can perform. We focus on
computing answers to a fundamental class of database queries: the
\emph{linear queries}, which generalize counting queries.  A counting
query counts the number of database elements that satisfy a given
predicate; a linear query allows for weighted
counts. Formally, a linear query is specified by a function
$q\maps\univ \to \R$ ($q\maps \univ \to \{0, 1\}$ in the case of
counting queries); slightly abusing notation, we define the value of
the query as $q(D) \triangleq \sum_{\row \in D}{q(e)}$ (elements of
$D$ are counted with multiplicity). We call a set $\quer$ of linear
queries a \emph{workload}, and an algorithm that answers a query
workload a \emph{mechanism}. 

Since the work of Dinur and
Nissim~\cite{DN}, it has been known that answering queries too
accurately can lead to very dramatic privacy breaches, and this is
true even for counting queries. For example, in~\cite{DN,DMT07} it was
shown that answering $\Omega(n)$ random counting queries with error
per query $o(\sqrt{n})$ allows an adversary to reconstruct a very
accurate representation of a database of size $n$, which contradicts
any reasonable privacy notion. On the other hand, a simple mechanism
that adds independent Gaussian noise to each query answer achieves
$(\eps, \delta)$-differential privacy and answers any set $\quer$ of
counting queries with average error
$O(\sqrt{|\quer|})$~\cite{DN,DworkN04,DMNS}.\footnote{Here and in the
  remainder of the introduction we ignore dependence of the error on
  $\eps$ and $\delta$.} While this is a useful guarantee for a small
number of queries, it quickly loses value when $|\quer|$ is much
larger than the database size, and becomes trivial for $\omega(n^2)$
queries. Nevetheless, since the seminal paper of Blum, Ligett and
Roth~\cite{BLR}, a long line of
work~\cite{DworkNRRV09,DworkRV10,RothR10,HardtR10,GuptaHRU11,HardtLM12,GuptaRU11}
has shown that even when $|\quer| = \omega(n)$, more sophisticated
private mechanisms can achieve error not much larger than
$O(\sqrt{n})$. For instance, there exist $(\eps,
\delta)$-differentially private mechanisms for linear queries that
acheive average error
$O(\sqrt{n}\log^{1/4}|\univ|)$~\cite{GuptaRU11}. There are sets of
counting queries for which this bound is tight up to factors
polylogarithmic in the size of the database~\cite{BunUV13}.

Specific query workloads allow for error which is much better than the
worst-case bounds. Some natural examples are queries counting the
number of points in a line interval or a $d$-dimensional axis-aligned
box~\cite{dwork-continual,ChanSS10,XiaoWG10}, or a $d$-dimensional
halfspace~\cite{halfspaces}. It is, therefore, desirable to have
mechanisms whose error bounds adapt \emph{both} to the query workload
and to the database size. In particular, if $\opt(n, \quer)$ is the
best possible average error\footnote{We give a formal definition later.}
achievable under differential privacy for the workload $\quer$ on
databases of size at most $n$, we would like to have a mechanism with
error at most a small factor larger than $\opt(n,\quer)$ for any $n$
and $\quer$. The first result of this type is due to Nikolov, Talwar,
and Zhang~\cite{NTZ}, who presented a mechanism running in time
polynomial in $|\univ|$, $|\quer|$, and $n$, with error at most
$\polylog(|\quer|, |\univ|)\cdot\opt(n,\quer)$.

Here we improve the results from~\cite{NTZ}:
\begin{theorem}[Informal]\label{thm:main}
  There exists a mechanism that, given a database of size $n$ drawn
  from a universe $\univ$, and a workload $\quer$ of linear queries,
  runs in time polynomial in $|\univ|$, $|\quer|$ and $n$, and has
  average error per query at most $\polylog(n, |\univ|)\cdot \opt(n,\quer)$.
\end{theorem}

Notice that the competitiveness ratio in Theorem~\ref{thm:main} is
\emph{independent of the number of queries}, which can be
significantly larger than both $n$ and $|\univ|$. This type of
guarantee is easier to prove when $n = \Omega(|\quer|)$, because in that case
there exist nearly optimal mechanisms that are oblivious of the
database size~\cite{NTZ}. Therefore, we focus on the more
challenging regime of small databases, i.e.~$n = o(|\quer|)$. 

It is worth making a couple of remarks about the strength of
Theorem~\ref{thm:main}. First, in many applications the queries in
$\quer$ are represented compactly rather than by a truth table, and
$|\univ|$ is exponentially large in the size of a natural
representation of the input. In such cases, running time bounds which
are polynomial in $|\univ|$ may be prohibitive. Nevertheless, our work
still gives interesting information theoretic bounds on the optimal
error, and, moreover, our mechanism can be a starting point for
developing more efficient variants. Furthermore, under a plausible
complexity theoretic hypothesis, our running time guarantee is the
best one can hope for without making further assumptions on
$\quer$~\cite{Ullman13}. A second remark is that our optimal error
guarantees are in terms of \emph{average} error, while many papers in
the literature consider worst-case error. Proving a result analogous
to Theorem~\ref{thm:main} for worst-case error remains an interesting
open problem.

\junk{Another
interesting problem is to remove the dependence on the universe size
in the competitiveness ratio. It is plausible that this can be done
with the projection mechanism and a well-chosen Gaussian noise
distribution, but we would need tighter lower bounds, possibly based
on fingerprinting codes as in~\cite{BunUV13}.}




\paragraph{Techniques.}
Following the ideas of~\cite{NTZ}, our starting point is a
generalization of the well-known Gaussian noise mechanism, which adds
appropriately scaled correlated Gaussian noise to the queries. By
itself, this mechanism is sufficient to guarantee privacy, but its
error is too large when $n = o(|\quer|)$. The main insight
of~\cite{NTZ} was to use the knowledge that the database is small to
reduce the error via a post-processing step. The post-processing is a
form of regression: we find the vector of answers that is closest to the
noisy answers while still consistent with the database size bound.
(In fact the estimator is slightly more complicated and related to the
hybrid estimator of Zhang~\cite{Zhang13-hybrid}.)  Intuitively, when
$n$ is small compared to the number of queries, this regression step
cancels a significant fraction of the error.

Our first novel contribution is to analyze the error of this mechanism
for arbitrary noise distributions and formulate it as a convex
function of the covariance matrix of the noise. Then we write a convex
program that captures the problem of finding the covariance matrix for
which the performance of the mechanism is optimized on the given query
workload and database size bound. We use Gaussian noise with this
optimal covariance in place of the recursively constructed ad-hoc
noise distribution\footnote{The distribution in~\cite{NTZ} is
  independent of the database size bound. This could be a reason why
  their guarantees scale with $\log |\quer|$ rather than $\log n$.}
from~\cite{NTZ}. Finally, we relate the dual of the convex program to
a spectral lower bound on $\opt(n, \quer)$ via the restricted
invertibility principle of Bourgain and Tzafriri~\cite{bour-tza}. We
stress that while the restricted invertibility principle was used
in~\cite{NTZ} as well, here we need a new argument which works for the
optimal covariance matrix we compute and gives a smaller
competitiveness ratio.

In addition to the improvement in the competitiveness ratio, our
approach here is more direct and we believe it gives a better
understanding of the performance of the regression-based mechanism for
small databases.

\section{Preliminaries}

We use capital letters for matrices and lower-case letters for vectors
and scalars. We use $\langle \cdot, \cdot \rangle$ for the standard
inner product between vectors in $\R^n$. For a matrix $M \in
\R^{m\times n}$ and a set $S \subseteq [n]$, we use $M_S$ for the
submatrix consisting of the columns of $A$ indexed by elements of
$S$. We use the notation $M \succ 0$ to denote that $M$ is a positive
definite matrix, and $M\succeq 0$ to denote that it is positive
semidefinite. We use $\sigma_{\min}(M)$ for the smallest singular
value of $M$, i.e.~$\sigma_{\min}(M) \triangleq
\min_{x}{\|Mx\|_2/\|x\|_2}$. We use $\tr(\cdot)$ for the trace operator, and
$\|M\|_2$ for the $\ell_2\to\ell_2$ operator norm of $M$,
i.e.~$\|M\|_2 \triangleq \max_{x}{\|Mx\|_2/\|x\|_2}$.

The distribution of a multivariate Gaussian
with mean $\mu$ and covariance $\Sigma$ is denoted $N(\mu, \Sigma)$.

\subsection{Histograms, the Query Matrix, and the Sensitivity Polytope}

It will be convenient to encode the problem of releasing answers to
linear queries using linear-algebraic notation. A common and very
useful representation of a database $D$ is the
\emph{histogram representation}: the histogram of $D$ is a vector $x
\in \R^\univ$ such that for any $\row \in \univ$, $x_\row$ is equal to
the number of copies of $\row$ in $D$. Notice that $\|x\|_1 = n$ and
also that if $x$ and $x'$ are respectively the histograms of two
neighboring databases $D$ and $D'$, then $\|x - x'\|_1 \leq 1$ (here
$\|x\|_1 = \sum_{\row}{|x_\row|}$ is the standard $\ell_1$
norm). Linear queries are a linear transformation of $x$. More
concretely, let us define the \emph{query matrix} $A \in \R^{\quer
  \times \univ}$ associated with a set of linear queries $\quer$ by
$a_{q, \row} = q(\row)$. Then it is easy to see that the vector $Ax$
gives the answers to the queries $\quer$ on a database $D$ with
histogram $x$.

Since this does not lead to any loss in generality, for the remainder
of this chapter we will assume that databases are given to mechanisms
as histograms, and workloads of linear queries are given as query
matrices. We will identify the space of size-$n$ databases with
histograms in the scaled $\ell_1$ ball $nB_1^\univ \triangleq \{x\in
\R^\univ: \|x\|_1 \leq n\}$, and we will identify neighboring
databases with histograms $x,x'$ such that $\|x-x'\|_1 \leq 1$.

The \emph{sensitivity polytope} $K_A$ of a query matrix $A \in
\R^{\quer\times \univ}$ is the convex hull of the columns of $A$ and
the columns of $-A$. Equivalently, $K_A \triangleq AB_1^\univ$,
i.e.~the image of the unit $\ell_1$ ball in $\R^\univ$ under
multiplication by $A$. Notice that $nK_A = \{Ax: \|x\|_1 \leq n\}$ is
the symmetric convex hull\footnote{The symmetric convex hull of a set
  of points $v_1,\ldots,v_N$ is equal to the convex hull of $\pm v_1,
  \ldots, \pm v_N$.} of the possible vectors of query answers to the
queries in $\quer$ on databases of size at most $n$.

\subsection{Measures of Error and the Spectral Lower Bound}

As our basic notion of error we will consider mean squared error.
For a mechanism $\mech$ and a subset $X\subseteq \R^\univ$, let us
define the error with respect to the query matrix $A \in
\R^{\quer\times \univ}$ as
\begin{equation*}
  \err(\mech, X, A) \triangleq \sup_{x \in X} \Biggl(\E\frac{1}{|\quer|} \|Ax -
    \mech(A, x)\|_2^2\Biggr)^{1/2}.
\end{equation*}
where the expectation is taken over the random coins of
$\mech$.
We also write $\err(\mech,nB_1^\univ, A)$ as $\err(\mech, n, A)$. The optimal
error achievable by any $(\eps, \delta)$-differentially private
mechanism for the query matrix $A$ and databases of size up to $n$ is
\begin{equation*}
  \opt_{\eps, \delta}(n, A) \triangleq \inf_{\mech} \err(\mech, n, A),
\end{equation*}
where the infimum is taken over all $(\eps, \delta)$-differentially
private mechanisms $\mech$.

Arguing directly about $\opt_{\eps, \delta}(n, A)$ appears
difficult. For this reason we use the following spectral lower bound
from~\cite{NTZ}. This lower bound was implicit in previous papers, for example~\cite{shiva2010}.

\begin{theorem}[\cite{NTZ}]\label{thm:speclb-small}
  There exists a constant $c$ such that for any query matrix $A \in
  \R^{\quer \times \univ}$, any small enough $\eps$, and any $\delta$
  small enough with respect to $\eps$, $\opt_{\eps, \delta}(n, A) \geq
  (c/\eps) \specLB(\eps n, A)$, where
  \begin{equation*}
    \specLB(k, A) \triangleq \max_{\substack{S \subseteq \univ\\|S| \leq
        k}} {\sqrt{{k}/{|\quer|}}\ \sigma_{\min}(A_S)}.
  \end{equation*}

\end{theorem}

\subsection{Composition and the Gaussian Mechanism}

An important basic property of differential privacy is that the
privacy guarantees degrade smoothly under composition and are not
affected by post-processing.

\begin{lemma}[\cite{DMNS,odo}]
  \label{lm:simple-composition}
  Let $\mech_1(\cdot)$ satisfy $(\eps_1, \delta_1)$-differential privacy, and
  $\mech_2(x,\cdot) $ satisfy $(\eps_2, \delta_2)$-differential
  privacy for any fixed $x$. Then the mechanism
  $\mech_2(\mech_1(D), D)$ satisfies
  $(\eps_1 + \eps_2, \delta_1 + \delta_2)$-differential privacy.
\end{lemma}

A basic method to achieve $(\eps, \delta)$-differential privacy is the
Gaussian mechanism. We  use the following generalized variant,
introduced in~\cite{NTZ}.

\begin{theorem}[\cite{DN,DworkN04,DMNS,NTZ}]\label{thm:gaussian}
  Let $\quer$ be a set of queries with query matrix $A$, and let
  $\Sigma \in \R^{\quer \times \quer}$, $\Sigma\succ 0$, be such that
  $a_e^T\Sigma^{-1}a_e \leq 1$ for all columns $a_e$ of $A$. Then the
  mechanism $\mech_\Sigma(A, x) = Ax + w$ where $w \sim N(0,c_{\eps,
    \delta}^2\Sigma)$ and $c_{\eps, \delta}\triangleq \frac{0.5\sqrt{\eps} +
    \sqrt{2\ln(1/\delta)}}{\eps}$ satisfies $(\eps,
  \delta)$-differential privacy.
\end{theorem}

\section{The Projection Mechanism}

A key element in our mechanism is the use
of least squares estimation to reduce error on small databases. In this section we
introduce and analyze a mechanism based on least squares estimation,
similar to the hybrid estimator of~\cite{Zhang13-hybrid}. Essentially
the same mechanism was used in~\cite{NTZ}, but the definition and
analysis were tied to a particular noise distribution.

\begin{algorithm}
  \caption{Projection Mechanism    $\alg_{\Sigma}^{\text{proj}}$}\label{alg:proj}
  \begin{algorithmic}[1]
    \REQUIRE \emph{(Public)} Query matrix $A \in \R^{\quer \times
      \univ}$; matrix $\Sigma \succ 0$ such that $a_e^\tra \Sigma^{-1}
      a_e \leq$ for all columns $a_e$ of $A$.
    \REQUIRE \emph{(Private)} Histogram $x$ of a database of size
    $\|x\|_1 \leq n$.
    
    \STATE Run the generalized Gaussian mechanism
    (Theorem~\ref{thm:gaussian}) to compute $\tilde{y}
    \triangleq \mech_\Sigma(A,x)$;

    \STATE Let $\Pi$ be the orthogonal projection operator onto the
    span of the eigenvectors corresponding to the  $\lfloor \eps n
    \rfloor$ largest eigenvalues of $\Sigma$

    \STATE Compute $\bar{y} \in n(I-\Pi)K_A$, where $K_A$ is the sensitivity
    polytope of $A$, and $\bar{y}$ is
    \[
    \bar{y} = \arg\min\{\|z - (I-\Pi)\tilde{y}\|_2^2: z \in n(I-\Pi)K_A\}.
    \]

    \ENSURE  Vector of answers $\Pi \tilde{y} + \bar{y}$.
  \end{algorithmic}
\end{algorithm}

As shown in~\cite{NTZ,conjunctions}, Algorithm~\ref{alg:proj} can be
efficiently implemented using the ellipsoid algorithm or the
Frank-Wolfe algorithm.

To analyze the error of the Projection Mechanism, we use the following
key lemma, which appears to be standard in statistics (we refer
to~\cite{NTZ,conjunctions} for a proof). Recall that for a convex body
(compact convex set with non-empty interior) $L \subseteq \R^m$, the
\emph{Minkowski norm} (\emph{gauge function}) is defined by $\|x\|_L \triangleq
\min\{r: x\in rL\}$ for any $x \in \R^m$. The \emph{polar body} is $L^\circ
\triangleq \{y: \langle y, x\rangle \leq 1\ \forall x \in L\}$ and the
corresponding norm is also equal to the \emph{support function} of
$L$: $\|y\|_{L^\circ} \triangleq \max\{\langle y, x\rangle: x \in
L\}$. When $L$ is symmetric around $0$ (i.e.~$-L = L$), the Minkowski norm
and support function are both norms in the usual sense.

\begin{lemma}[\cite{NTZ,conjunctions}]\label{lm:lse}
  Let $L \subseteq \R^m$ be a symmetric convex body, and let $y\in L,
  \tilde{y}\in \R^m$. Let, finally, $\bar{y} =\arg\min\{\|z -
  \tilde{y}\|_2^2: z \in L\}$. We have $\|\bar{y} - y\|_2^2 \leq
  4\min\{\|\tilde{y} - y\|_2^2, \|\tilde{y} - y\|_{L^\circ}\}.  $
\end{lemma}

The next lemma gives our analysis of the error of the Projection
Mechanism.

\begin{lemma}\label{lm:proj-err}
  Assume $\Sigma \succ 0$ is such that $a_e^\tra \Sigma^{-1} a_e \leq
  1$ for all columns $a_e$ of $A$. Then the Projection Mechanism
  $\alg_\Sigma^{\text{proj}}$ in Algorithm~\ref{alg:proj} is $(\eps,
  \delta)$-differentially private. Moreover, for $\eps = O(1)$,
  \[
  \err(\alg_\Sigma^{\text{proj}}, n, A) =O\Biggl(\Bigl(1+\frac{\sqrt{\log |\univ|}}{\sqrt{\log
      1/\delta}}\Bigr)^{1/2}\Biggr)\cdot
  \Bigl(\frac{c^2_{\eps, \delta}}{|\quer|}\sum_{i \leq \eps n}{\sigma_i}\Bigr)^{1/2},
  \] where $\sigma_1 \ge \sigma_2\ge
  \ldots \ge \sigma_{|\quer|}$ are the eigenvalues of $\Sigma$.
\end{lemma}
\begin{proof} 
  To prove the privacy guarantee, observe that the output of
  $\alg_\Sigma^{\text{proj}}(A,x)$ is just a post-processing of the
  output of $\alg_\Sigma(A,x)$, i.e.~the algorithm does not access $x$
  except to pass it to $\alg_\Sigma(A,x)$. The privacy then follows
  from Theorem~\ref{thm:gaussian} and
  Lemma~\ref{lm:simple-composition}.

  Next we bound the error. Let $y\triangleq Ax$ be the true answers to
  the queries, and let $w \triangleq \tilde{y} - y \sim N(0, c_{\eps,
    \delta}^2\Sigma)$ be the random noise introduced by the
  generalized Gaussian mechanism.  By the Pythagorean theorem and
  linearity of expectation, the expected total squared error of the
  projection mechanism is
  \[
  \E\|\Pi \tilde{y} + \bar{y} - y\|_2^2 =   \E\|\Pi \tilde{y}  - \Pi
  y\|_2^2 +\E \|\bar{y} - (I-\Pi) y\|_2^2.
  \]
  Above and in the remainder of the proof expectations are taken
  with respect to the randomness in the choice of $w$.  We bound the
  two terms on the right hand side separately. We will show:
  \begin{align}
    \label{eq:top-eig}
    \E\|\Pi \tilde{y}  - \Pi  y\|_2^2 &=   c_{\eps, \delta}^2\sum_{i =
      1}^k{\sigma_i},\\
    \label{eq:bott-eig}
    \E \|\bar{y} - (I-\Pi) y\|_2^2 &= O\left(\frac{\sqrt{\log
          |\univ|}}{\sqrt{\log 1/\delta}}\right)c_{\eps, \delta}^2\sum_{i =
      1}^k{\sigma_i}.
  \end{align}
  \eqref{eq:top-eig} and \eqref{eq:bott-eig} together imply the error
  bound in the theorem.

  To prove \eqref{eq:top-eig}, observe that $\Pi\tilde{y} - \Pi y =
  \Pi w \sim N(0, c_{\eps,\delta}^2\Pi \Sigma \Pi)$. By the definition of $\Pi$, the
  non-zero eigenvalues of $\Pi \Sigma\Pi$ are $\sigma_1, \ldots,
  \sigma_{k}$ where $k \triangleq \lfloor \eps n\rfloor$. We have
  \[
  \E \|\Pi\tilde{y} - \Pi y\|_2^2 = c_{\eps, \delta}^2\tr(\Pi \Sigma \Pi) =
  c_{\eps, \delta}^2\sum_{i = 1}^k{\sigma_i}.
  \]

  To prove \eqref{eq:bott-eig} we appeal to Lemma~\ref{lm:lse}. Define
  $\tilde{K} \triangleq (I-\Pi)K_A$. With $n\tilde{K}$ in the place of
  $L$, the lemma implies that
  \begin{equation}\label{eq:err-bnd-lse}
  \E \|\bar{y} - (I-\Pi)y\|_2^2 \leq 4\E\|(I-\Pi)w\|_{(n\tilde{K})^\circ}
  = 4n\E\|(I-\Pi)w\|_{\tilde{K}^\circ},
  \end{equation}
  where we used the simple fact
  \[
  \|(I-\Pi)w\|_{(n\tilde{K})^\circ} = \sup_{z \in n\tilde{K}}{\langle
    (I-\Pi)w, z \rangle} = n\sup_{z \in \tilde{K}}{\langle (I-\Pi)w, z
    \rangle} = n\|(I-\Pi)w\|_{\tilde{K}^\circ}.\]
  $\tilde{K}$ is the convex hull of the columns of $(I-\Pi)A$ and the
  columns of $-(I-\Pi)A$. For any such column $(I-\Pi)a_e$ we have
  \[
    1 \geq a_e^\tra \Sigma^{-1} a_e \geq a_e^\tra(I-\Pi)\Sigma^{-1}(I-\Pi) a_e \geq
    \sigma_{k+1}^{-1}a_e^\tra(I - \Pi)a_e.
  \]
  The first inequality is by the assumption on $\Sigma$; the second
  follows because $\Sigma^{-1} - (I-\Pi)\Sigma^{-1}(I-\Pi) \succeq 0$;
  the third inequality is due to the fact that the smallest eigenvalue
  of $(I-\Pi)\Sigma^{-1}(I-\Pi)$ restricted to the range of $I-\Pi$ is
  $\sigma_{k+1}^{-1}$ by the choice of $\Pi$. Therefore,
  $\|(I-\Pi)a_e\|_2^2 \leq \sigma_{k+1}\leq \sigma_k$. Since a linear
  functional attains its maximum value over a polytope at a vertex, we
  have $\|(I-\Pi)w\|_{\tilde{K}^\circ} = \sup_{z \in
    \tilde{K}}{\langle (I-\Pi)w,z\rangle} = \max_{e\in \univ}{|\langle
    (I-\Pi)w, a_e\rangle}|$. Each inner product $\langle (I-\Pi)w,
  a_e\rangle$ is a centered Gaussian random variable with variance $
  \E\langle (I-\Pi)w, a_\row\rangle^2 = c_{\eps, \delta}^2a_\row^\tra
  (I-\Pi)\Sigma(I-\Pi)a_\row.  $ By the choice of $\Pi$, the largest
  eigenvalue of $(I-\Pi)\Sigma(I-\Pi)$ is $\sigma_{k+1} \leq
  \sigma_k$. From this fact and the inequality $\|(I-\Pi)a_e\|_2^2
  \leq \sigma_k$, we have that the variance of $\langle (I-\Pi)w,
  a_\row\rangle$ is at most $c_{\eps, \delta}^2\sigma_k^2$. By a
  standard concentration argument, we can bound the expectation of the
  maximum absolute value of the inner products as
  \[
  \E\|(I-\Pi)w\|_{\tilde{K}^\circ} =\E \max_{\row \in \univ}{|\langle
    (I-\Pi)w, a_\row\rangle|} = O(\sqrt{\log |\univ|})c_{\eps, \delta}\sigma_k.
  \]
  Plugging this into \eqref{eq:err-bnd-lse}, we get
  \[
  \E \|\bar{y} - (I-\Pi)y\|_2^2 =  O(\sqrt{\log |\univ|})c_{\eps, \delta}n\sigma_k.
  \]
  To show that this implies \eqref{eq:bott-eig}, observe that, by
  averaging, $c_{\eps, \delta}n\sigma_k \leq \frac{c_{\eps,
      \delta}n}{k} \sum_{i=1}^k{\sigma_i}$.  Since $k = \lfloor \eps n
  \rfloor$, $\frac{c_{\eps, \delta}n}{k} = O\Bigl(\frac{c_{\eps,
      \delta}^2}{\sqrt{\log 1/\delta}}\Bigr)$. This finishes the
  proof of \eqref{eq:bott-eig}, and, therefore, of the theorem.
\end{proof}

\section{Optimality of the Projection Mechanism}

In this section we show that we can choose a covariance matrix
$\Sigma$ so that $\alg_\Sigma^{\text{proj}}$ has nearly optimal
error:
\begin{theorem}\label{thm:main-smalldb}
  Let $\eps$ be a small enough constant and let $\delta =
  |\univ|^{o(1)}$ be small enough with respect to $\eps$. For any
  query matrix $A \in \R^{\quer \times \univ}$, and any database size
  bound $n$, there exists a covariance matrix $\Sigma \succ 0$ such
  that the Projection Mechanism $\alg_\Sigma^{\text{proj}}$ in
  Algorithm~\ref{alg:proj} is $(\eps, \delta)$-differentially private
  and has error
  \begin{align*}
  \err(\alg,n, A) &=
  O((\log n) (\log 1/\delta)^{1/4} (\log |\univ|)^{1/4})
  \cdot\frac{1}{\eps}\specLB(\eps n, A)\\
  &= O((\log n) (\log 1/\delta)^{1/4} (\log |\univ|)^{1/4}) \cdot\opt_{\eps, \delta}(n, A)
  \end{align*}
  Moreover, $\Sigma$ can be computed in time polynomial in $|\quer|$.
\end{theorem}

Theorem~\ref{thm:main-smalldb} is the formal statement of
Theorem~\ref{thm:main}. (Recall again that Algorithm~\ref{alg:proj}
can be implemented in time polynomial in $n$, $|\quer|$ and $|\univ|$, as shown in~\cite{NTZ,conjunctions}.)

To prove the theorem, we optimize over the choices of $\Sigma$ that
ensure $(\eps, \delta)$-differential privacy, and use convex duality
and the restricted invertibility principle to relate the optimal
covariance to the spectral lower bound.

\subsection{Minimizing the Ky Fan Norm}

Recall that for an $m\times m$ matrix $\Sigma \succ 0$ with
eigenvalues $\sigma_1 \geq \ldots \geq \ldots \geq \sigma_m$, and a
positive integer $k \leq m$, the Ky Fan $k$-norm is defined as
$\|\Sigma\|_{(k)} \triangleq \sigma_1 + \ldots + \sigma_k$.  The
covariance matrix $\Sigma$ we use in the projection mechanism will be
the one achieving $\min\{\|\Sigma\|_{(k)}: a_e^\tra \Sigma^{-1}a_e \leq
1~\forall e\in \univ\}$, where $a_e$ is the column of the query matrix
$A$ associated with the universe element $e$. This choice is directly
motivated by Lemma~\ref{lm:proj-err}.  We can write this optimization
problem in the following way.
\begin{align}
  &\text{Minimize } \|X^{-1}\|_{(k)} \label{eq:ellips-obj-small}
  \text{ s.t. }\\
  &X\succ 0\\
  &\forall e \in \univ: a_e^\intercal Xa_e \leq 1.\label{eq:ellips-enclose-small}
\end{align}
The program above has a geometric meaning. For a positive definite
matrix $X$, the set $E(X)\triangleq\{v \in \R^\quer:
v^\tra X v\}$ is an ellipsoid centered at the origin. The constraint
\eqref{eq:ellips-enclose-small} means that $E(X)$ has to contain all
columns of the query matrix $A$. The objective function
\eqref{eq:ellips-obj-small} is equal to the sum of squared lengths of
the $k$ longest major axes of $E(X)$. Therefore, we are looking for
the smallest  ellipsoid centered at the origin that contains the
columns of $A$, where the ``size'' of the ellipsoid is the sum of
squared lengths of the $k$ longest major axes. We will not use this
geometric interpretation in the rest of the paper.

We will show that
\eqref{eq:ellips-obj-small}--\eqref{eq:ellips-enclose-small} is a
convex optimization problem. This will allow us to use general tools
such as the ellipsoid method to find an optimal solution, and also to
use duality theory in order to analyze the value of the optimal
solution.

To show that
\eqref{eq:ellips-obj-small}--\eqref{eq:ellips-enclose-small} is convex
we will need the following well-known result of Fan.

\begin{lemma}[\cite{Fan49}]\label{lm:kyfan}
  For any $m\times m$ real symmetric matrix $\Sigma$, 
  \[
  \|\Sigma\|_{(k)} = \max_{U \in \R^{m\times k}: U^\tra U = I} \tr(U^\tra \Sigma U).
  \]
\end{lemma}

With this result in hand, we can prove that
\eqref{eq:ellips-obj-small}--\eqref{eq:ellips-enclose-small} is a
convex optimization problem.

\begin{lemma}\label{lm:ellips-program-small}
  The objective function \eqref{eq:ellips-obj-small} and constraints
  \eqref{eq:ellips-enclose-small} are convex over $X \succ 0$.
\end{lemma}
\begin{proof}
  The objective function and the constraints
  \eqref{eq:ellips-enclose-small} are affine, and therefore convex. It
  remains to show that the objective \eqref{eq:ellips-obj-small} is
  also convex. Let $X_1$ and $X_2$ be two feasible solutions and
  define $Y = \alpha X_1 + (1-\alpha)X_2$ for some $\alpha \in [0,
  1]$. Because the matrix inverse is operator convex (see
  e.g.~\cite{Bhatia-MA}), $Y^{-1} \preceq \alpha X_1^{-1} + (1-\alpha)
  X_2^{-1}$. Let $U\in \R^{m\times k}$ be such that $\tr(U^\tra Y^{-1} U) =
  \|Y^{-1}\|_{(k)}$ and $U^\tra U = I$. Such a $U$ exists by by
  Lemma~\ref{lm:kyfan}. We have, again using Lemma~\ref{lm:kyfan},
  \begin{align*}
  \|Y^{-1}\|_{(k)} = \tr(U^\tra Y^{-1} U) &\leq \alpha \tr(U^\tra
  X_1^{-1}U) +  (1-\alpha) \tr(U^\tra X_2^{-1}U)\\
  &\leq  \alpha\|X_1^{-1}\|_{(k)} + (1-\alpha)\|X_2^{-1}\|_{(k)}. 
  \end{align*}
  This finishes the proof. 
\end{proof}

Since the program
\eqref{eq:ellips-obj-small}--\eqref{eq:ellips-enclose-small} is
convex, its optimal solution can be approximated in polynomial time within
any given degree of accuracy using the ellipsoid
algorithm~\cite{GLS-ellipsoid}.

\subsection{A Special Function}

Before we continue, we need to introduce a somewhat complicated
function of the singular values of a matrix. This function will turn
out to be the objective funciton in a maximization problem which is
dual to \eqref{eq:ellips-obj-small}--\eqref{eq:ellips-enclose-small}.
The next lemma is needed to argue that this function is
well-defined. The lemma was proved in~\cite{simplex}.

\begin{lemma}[\cite{simplex}]\label{lm:singvals-thresh}
  Let $\sigma_1 \geq \ldots \sigma_m \geq 0$ be non-negative reals,
  and let $k \leq m$ be a positive integer. There exists a unique integer
  $t$, $0 \leq t \leq k-1$, such that
  \begin{equation}\label{eq:singvals-thresh}
  \sigma_t > \frac{\sum_{i > t}{\sigma_i}}{k - t} \geq \sigma_{t+1},
  \end{equation}
  with the convention $\sigma_0 = \infty$. 
\end{lemma}

We are now ready to define the function:
\begin{definition}
  Let $\Sigma \succeq 0$ be an $m\times m$ positive semidefinite
  matrix with singular values $\sigma_1\geq \ldots \geq \sigma_m$, and
  let $k \leq m$ be a positive integer. The function $h_k(\Sigma)$ is
  defined as
  \[
  h_k(\Sigma) \triangleq \sum_{i = 1}^t{\sigma_i^{1/2}} +  \sqrt{k-t}\left(\sum_{i > t}{\sigma_i}\right)^{1/2},
  \]
  where $t$ is the unique integer such that 
  $
  \sigma_t > \frac{\sum_{i > t}{\sigma_i}}{k - t} \geq \sigma_{t+1}.
  $
\end{definition}
Lemma~\ref{lm:singvals-thresh} guarantees that $h_k(\Sigma)$ is a
well-defined real-valued function. In the next lemma we also show that
it is continuous.

\begin{lemma}\label{lm:hk-contin}
  The function $h_k$ is continuous over positive semidefinite
  matrices with respect to the operator norm.  
\end{lemma}
\begin{proof}
  Let $\Sigma$ be a $m\times m$ positive semidefinite matrix with singular
  values $\sigma_1 \geq \ldots \geq \sigma_m$ and let $t$, $0\leq t < k$, be the
  unique integer for which $\sigma_t > \frac{\sum_{i > t}{\sigma_i}}{k
    - t} \geq \sigma_{t+1}$. If $\frac{\sum_{i > t}{\sigma_i}}{k - t}
  > \sigma_{t+1}$, then setting $\delta$ small enough ensures that, for
  any $\Sigma'$ such that $\|\Sigma - \Sigma'\|_2 < \delta$, $h_k(\Sigma)$ and $h_k(\Sigma')$ are
  computed with the same value of $t$. In this case, the proof of
  continuity follows from the continuity of the square root
  function. Let us therefore assume that $\frac{\sum_{i > t}{\sigma_i}}{k - t}
  = \sigma_{t+1} = \ldots = \sigma_{t'}>\sigma_{t'+1}$ for some $t'
  \geq t+1$. Then  for any integer $s \in [t, t']$,
  \[
  \sum_{i > s}{\sigma_i} = \sum_{i>t}{\sigma_i} - (s-t)\sigma_{t+1}
  = (k-s)\sigma_{t+1}.
  \]
  We then have
  \begin{align}\label{eq:hk-irrev-t}
    \sum_{i = 1}^t{\sigma_i^{1/2}} +  \sqrt{k-t}\left(\sum_{i >
        t}{\sigma_i}\right)^{1/2} &= 
    \sum_{i = 1}^t{\sigma_i^{1/2}} +  (k-t)\sigma^{1/2}_{t+1}\notag\\
    &= \sum_{i = 1}^s{\sigma_i^{1/2}} +
    (k-s)\sigma^{1/2}_{t+1}\notag\\
    &= \sum_{i = 1}^s{\sigma_i^{1/2}} +  \sqrt{k-s}\left(\sum_{i > s}{\sigma_i}\right)^{1/2}.
  \end{align}
  For any $\Sigma'$ such that $\|\Sigma' - \Sigma\|_2 < \delta$ for a small enough
  $\delta$, we have
  \begin{equation*}
  h_k(\Sigma') =  \sum_{i = 1}^s{\sigma_i(\Sigma')^{1/2}} +  \sqrt{k-s}\left(\sum_{i >
     s}{\sigma_i(\Sigma')}\right)^{1/2},
  \end{equation*}
for some integer $s$  in $[t,t']$. Continuity then follows from
  \eqref{eq:hk-irrev-t}, and the
  continuity of the square root function. 
\end{proof}

\subsection{The Dual of the Ky Fan Norm Minimization Problem}

Our next goal is derive a dual characterization of
\eqref{eq:ellips-obj-small}--\eqref{eq:ellips-enclose-small}, which we
will then relate to the spectral lower bound $\specLB(k, A)$. It is
useful to work with the dual, because it is a maximization problem, so
to prove optimality we just need to show that any feasible solution of
the dual gives a lower bound on the optimal error under differential
privacy.

The next theorem gives our dual characterization in terms of the
special function $h_k$ defined in the previous section.

\begin{theorem}\label{thm:nuclear-small}
  Let $A = (a_e)_{e \in \univ}\in \R^{\quer\times \univ}$ be a rank
  $|\quer|$ matrix, and let $\mu$ be the optimal value of
  \eqref{eq:ellips-obj-small}--\eqref{eq:ellips-enclose-small}. Then,
  \begin{align}
    \mu^2 =   &\max h_k(AQA^\tra)^2\label{eq:nuclear-obj2-small}
    \text{ s.t.}\\
    &Q \succeq 0, \text{ diagonal},\tr(Q) = 1\label{eq:nuclear-pos2-small}
  \end{align}
\end{theorem}

Since the objective of
\eqref{eq:ellips-obj-small}--\eqref{eq:ellips-enclose-small} is not
necessarily differentiable, in order to analyze the dual and prove
Theorem~\ref{thm:nuclear-small}, we need to recall the concepts of
subgradients and subdifferentials. A \emph{subgradient} of a convex
function $f\maps S \to \R$ at $x \in S$, where $S$ is some open subset
of $\R^d$, is a vector $y \in \R^d$ so that for every $z \in S$ we
have
\[
f(z) \geq f(x) + \langle z-x, y\rangle.
\]
The set of subgradients of $f$ at $x$ is denoted $\partial{f(x)}$ and
is known as the \emph{subdifferential}. When $f$ is differentiable at
$x$, the subdifferential is a singleton set containing only the
gradient $\nabla f(x)$. If $f$ is defined by $f(x) = f_1(x) + f_2(x)$,
where $f_1, f_2\maps S \to \R$ , then $\partial f(x) = \partial f_1(x)
+ \partial f_2(x)$. A basic fact in convex analysis is that $f$
achieves its minimum at  $x$ if and only if $0 \in \partial f(x)$. For
more information on subgradients and subdifferentials, see the
classical text of Rockafellar~\cite{Rockafellar}. 

Overton and Womersley~\cite{OvertonW93-kyfan} analyzed the
subgradients of functions which are a composition of a differentiable
matrix-valued function with a Ky Fan norm. The special case we need
also follows from the results of Lewis~\cite{Lewis95-matfunc}.
\begin{lemma}[\cite{OvertonW93-kyfan},\cite{Lewis95-matfunc}]\label{lm:kyfan-subgr}
  Let $g_k(X) \triangleq \|X^{-1}\|_{(k)}$ for a positive definite
  matrix $X \in \R^{m\times m}$. Let $\sigma_1 \ge \ldots \geq \sigma_m$ be the
  singular values of $X^{-1}$ and let $D$ be the diagonal matrix
  with the $\sigma_i$ on the diagonal. Assume that for some $r \geq
  k$, $\sigma_k = \ldots = \sigma_r$. Then the subgradients of $g_k$ are
  given by
  \[
  \partial g_k(X) = \conv\{-U_SU_S^\tra X^{-2} U_SU_S^\tra: U
  \text{orthonormal}, UD U^\tra = X^{-1}, S \subseteq [r]\},
  \]
  where $U_S$ is the submatrix of $U$ indexed by $S$. 
\end{lemma}

We use the following well-known characterization of the convex
hull of boolean vectors of Hamming weight $k$. For a proof,
see~\cite{schrijver-combop-B}. 

\begin{lemma}\label{lm:unik-poly}
  Let $V_{k,n} \triangleq \conv\{v \in \{0, 1\}^n: \|v\|_1 = k\}$. Then
  $V_{k,n} = \{v: \|v\|_1 = k, 0\leq v_i \leq 1~\forall i\}$. 
\end{lemma}

Before we prove Theorem~\ref{thm:nuclear-small}, we need one more
technical lemma.

\begin{lemma}\label{lm:subgr-soln}
  Let $\Sigma$ be an $m\times m$ positive semidefinite matrix of rank at
  least $k$. Then there exists an $m\times m$ positive definite matrix
  $X$ such that $\Sigma \in -\partial g_k(X)$, and $g_k(X) = \|X^{-1}\|_{(k)} = h_k(\Sigma)$. 
\end{lemma}
\begin{proof}
  Let $r = \rank~\Sigma$, and let $\sigma_1 \geq \ldots \geq \sigma_r$ be
  the non-zero singular values of $\Sigma$. Let $U D U^\tra = \Sigma$ be
  some singular value decomposition of $\Sigma$: $U$ is an orthonormal
  matrix and $D$ is a diagonal matrix with the $\sigma_i$ on the
  diagonal, followed by $0$s.

  Assume that $t$, $0 \leq t < k$, is the integer (guaranteed by
  Lemma~\ref{lm:singvals-thresh}) for which $ \sigma_t > \frac{\sum_{i
      > t}{\sigma_i}}{k - t} \geq \sigma_{t+1}$ and define $\alpha
  \triangleq \frac{\sum_{i > t}{\sigma_i}}{k - t}$. Since $t < k$ and
 we assumed $\Sigma$ has rank at least $k$, we have $\alpha > 0$.  Define
  \[
  \sigma'_{i} \triangleq
  \begin{cases}
    \sigma_{i} &i \leq t\\
    \alpha &t < i \leq r\\
    \alpha - \epsilon & i > r
  \end{cases},
  \]
  and set $D'$ be the diagonal matrix with the $\sigma'_i$ on the diagonal.
  We set $\epsilon$ to be an arbitrary number satisfying $\alpha >
  \epsilon > 0$. Let us set $X \triangleq (UD'U^\tra)^{-1/2}$. By
  Lemma~\ref{lm:unik-poly} and and the choice of $t$, the vector
  $(\sigma_{t+1}, \ldots, \sigma_r)$ is an element of the polytope
  $\alpha V_{k-t, r-t}$. Then $\Sigma$ is an element of $\conv\{U_SU_S^\tra
  X^{-2} U_S U_S^\tra: S = [t] \cup T, T \subseteq \{t+1, \ldots, r\},
  |T| = k-t\}$. Since this set is a subset of $-\partial g_k(X)$, we
  have $\Sigma \in -\partial g_k(X)$. A calculation shows that
  $\|X^{-1}\|_{(k)} = \|(UD'U^\tra)^{1/2}\|_{(k)} = \sum_{i \leq
    t}{\sigma_i^{1/2}} + (k-t)\alpha^{1/2} = h_k(\Sigma)$. This completes
  the proof. 
\end{proof}

\begin{proof}[of Theorem~\ref{thm:nuclear-small}]
  We will use standard notions from the theory of convex duality. For
  a reference, see the book by Boyd and Vandenberghe~\cite{BoydV-cvx}.

  Let us define $\{X: X \succ 0\}$ to be the domain for the
  constraints \eqref{eq:ellips-enclose-small} and the objective
  function \eqref{eq:ellips-obj-small}. This makes the constraint $X
  \succ 0$ implicit.  The optimization problem is convex by
  Lemma~\ref{lm:ellips-program-small}. Is is also always feasible: for
  example for $r$ an upper bound on the Euclidean norm of the longest
  column of $A$, $\frac{1}{r}I$ is a feasible
  solution. Slater's condition is therefore satisfied, since the
  constraints are affine, and, therefore, strong duality holds.

  The Lagrange dual function for
  \eqref{eq:ellips-obj-small}--\eqref{eq:ellips-enclose-small} is
  \begin{equation*}
    g(p) = \inf_{X \succ 0}{\|X^{-1}\|_{(k)} + \sum_{e \in \univ}{p_e(a_e^\intercal Xa_e - 1)} },
  \end{equation*}
  with dual variables $p \in \R^\univ$, $p \geq 0$. Equivalently, we
  can define the diagonal matrix $P \in \R^{\univ\times \univ}$, $P \succeq
  0$, with entries $p_{ee} = p_e$, and the dual function becomes
  \begin{equation}\label{eq:g-raw-small}
    g(P) = \inf_{X\succ 0}{\|X^{-1}\|_{(k)} + \tr(APA^\intercal X) - \tr(P)}
  \end{equation}
  Since the terms $\|X^{-1}\|_{(k)}$ and $\tr(APA^\intercal X)$ are
  non-negative for any $X \succ 0$, $g(P) \geq -\tr(P) >
  -\infty$. Therefore, the effective domain $\{P: g(P) > -\infty\}$ of
  $g(P)$ is $\{P: P\succeq 0, \text{ diagonal}\}$. Since we have
  strong duality, $\mu^2 = \max\{g(P): P\succeq 0, \text{ diagonal}\}$.
  
  By the additivity of subgradients, a matrix $X$  achieves the minimum
  in \eqref{eq:g-raw-small} if and only if $APA^\tra \in -\partial g_k(X)$,
  where $g_k(X) = \|X^{-1}\|_{(k)}$.  Consider first the case in which
  $APA^\tra$ has rank at least $k$. Then, by Lemma~\ref{lm:subgr-soln},
  there exists an $X$ such that $APA^\tra \in -\partial g_k(X)$ and
  $\|X^{-1}\|_{(k)} = h_k(APA^\tra)$. Observe that, if $U$ is an $m
  \times k$ matrix such
  that $U^\tra U = I$ and $\tr(U^\tra X^{-1} U) = \|X^{-1}\|_{(k)}$,
  then 
  \[
  \tr(UU^\tra X^{-2} UU^\tra X) = \tr((U^\tra X^{-2}U)(U^\tra X U))
  = \tr(U^\tra X^{-1} U) = \|X^{-1}\|_{(k)}.
  \]
  Since, by Lemma~\ref{lm:kyfan-subgr} and $APA^\tra \in -\partial g_k(X)$, $APA^\tra$ is a convex combination of matrices $UU^\tra X^{-2}  UU^\tra$ with $U$ as above, it follows that $\tr(APA^\tra X) =
  \|X^{-1}\|_{(k)}$. Then we have
  \begin{align}
  g(P) &= \|X^{-1}\|_{(k)} + \tr(APA^\intercal X) - \tr(P)\notag\\
  &= 2\|X^{-1}\|_{(k)} - \tr(P) = 2h_k(APA^\tra) - \tr(P).\label{eq:g-small-raw}
  \end{align}

  If $P$ is such that $APA^\tra$ has rank less than $k$, we can reduce
  to the rank $k$ case by a continuity argument. Fix any non-negative
  diagonal matrix $P$ and for $\lambda \in [0,1]$ define $P(\lambda)
  \triangleq \lambda P + (1-\lambda)I$. For any $\lambda \in [0, 1)$,
  $AP(\lambda)A^\tra$ has rank $|\quer|$, since $AA^\tra$ has rank
  $|\quer|$ by assumption, and, therefore, $AP(\lambda)A^\tra \succeq
  \lambda AA^\tra \succ 0$. Then, by Corollary
  7.5.1.~in~\cite{Rockafellar}, and \eqref{eq:g-small-raw}, we have
  \begin{align*}
  g(P) = \lim_{\lambda \uparrow 1}{g(P(\lambda))} &=
  \lim_{\lambda \uparrow 1}{[2h_k(AP(\lambda)A^\tra) - \lambda\tr(P)
  - (1-\lambda) |\quer|]} \\
  &=  2h_k(APA^\tra) - \tr(P).
  \end{align*} 
  The final equality follows from the continuity of $h_k$, proved in
  Lemma~\ref{lm:hk-contin}. 
  
  Let us define new variables $Q$ and $c$, where $c = \tr(P)$ and $Q =
  P/c$. Because $h_k$ is homogeneous with exponent $1/2$, we can
  re-write $g(P)$ as $ g(P) = g(Q,c) = 2\sqrt{c}h_k(AQA^\tra) - c$.
  From the first-order optimality condition $\frac{\partial
    g}{\partial c} = 0$, we see that maximum of $g(Q,c)$ is achieved
  when $c = h_k(AQA^\tra)^2$ and is equal to
  $h_k(AQA^\tra)^2$. Therefore maximizing $g(P)$ over diagonal
  positive semidefinite $P$ is equivalent to the optimization problem
  \eqref{eq:nuclear-obj2-small}--\eqref{eq:nuclear-pos2-small}. Since,
  by strong duality, the maximum of $g(P)$ is equal to the optimal
  value of
  \eqref{eq:ellips-obj-small}--\eqref{eq:ellips-enclose-small}, this
  completes the proof. 
\end{proof}

\subsection{Proof of  Theorem~\ref{thm:main-smalldb}}

Our strategy will be to use the dual formulation in
Theorem~\ref{thm:nuclear-small} and the restricted invertibility
principle to give a lower bound on $\specLB(k, A)$. First we state the
restricted invertiblity principle and a consequence of it proved in~\cite{apx-disc}.

\begin{theorem}[\cite{bour-tza,bt-constructive}]\label{thm:bt}
  Let $\epsilon \in (0,1)$, let $M$ be an $m\times n$ real matrix, and
  let $W$ be an $n\times n$ diagonal matrix such that $W \succeq 0$
  and $\tr(W) = 1$. For any integer $k$ such that $k \leq \epsilon^2
  {\tr(MWM^\tra)}/{\|MWM^\tra\|_2}$ there exists a subset
  $S \subseteq [n]$ of size $|S| = k$ such that
  $\sigma_{\min}(M_S)^2 \geq (1-\epsilon)^2\tr(MWM^\tra)$. 
\end{theorem}

For the following lemma, which is a consequence of
Theorem~\ref{thm:bt}, we need to recall the definition of the trace
(nuclear) norm of a matrix $M$: $\|M\|_{{\tr}}$ is equal to the
sum of singular values of $M$.

\begin{lemma}[\cite{apx-disc}]\label{lm:bt-lb}
  Let $M$ be an $m$ by $n$ real matrix of rank $r$, and let $W \succeq
  0$ be a diagonal matrix such that $\tr(W) = 1$. Then there exists a
  submatrix $M_S$ of $M$, $|S| \leq r$, such that
  $|S|\sigma_{\min}(M_S)^2 \geq {c^2\|MW^{1/2}\|_{\tr}^2}/{(\log
    r)^2}$, for a universal constant $c > 0$.
\end{lemma}

\begin{proof}[of Theorem~\ref{thm:main-smalldb}]
  Given a database size $n$ and a query matrix $A$, we compute the
  covariance matrix $\Sigma$ as follows. We compute a matrix $X$ which
  gives an (approximately) optimal solution to
  \eqref{eq:ellips-obj-small}--\eqref{eq:ellips-enclose-small} for  $k \triangleq
  \lfloor \eps n \rfloor$, and we
  set $\Sigma\triangleq X^{-1}$. Since
  \eqref{eq:ellips-obj-small}--\eqref{eq:ellips-enclose-small} is a
  convex optimization problem, it can be solved in time polynomial in
  $|\quer|$ to any degree of accuracy using the ellipsoid
  algorithm~\cite{GLS-ellipsoid} (or the algorithm of Overton and
  Womersley~\cite{OvertonW93-kyfan}). By Lemma~\ref{lm:proj-err} and
  the constraints \eqref{eq:ellips-enclose-small},
  $\mech_\Sigma^{\text{proj}}$ is $(\eps, \delta)$-differentially
  private with this choice of $\Sigma$.

  By Lemma~\ref{lm:proj-err},
  \begin{equation}\label{eq:proj-err}
    \err(\alg_\Sigma^{\text{proj}}, n, A) =O\Biggl(\Bigl(1+\frac{\sqrt{\log |U|}}{\sqrt{\log
      1/\delta}}\Bigr)^{1/2}\Biggr) \cdot \frac{c_{\eps,\delta}}{\sqrt{|\quer|}} \|\Sigma\|_{(k)}.
  \end{equation}
  By Theorem~\ref{thm:nuclear-small}, the optimal solution $Q$ of
  \eqref{eq:nuclear-obj2-small}--\eqref{eq:nuclear-pos2-small}
  satisfies
  \[
  \|\Sigma\|_{(k)} = h_k(AQA^\tra) = \sum_{i = 1}^t{\lambda_i^{1/2}} +
  \sqrt{k-t}\left(\sum_{i > t}{\lambda_i}\right)^{1/2},
  \]
  where $\lambda_1 \geq \ldots \geq \lambda_m$ are the eigenvalues
  of $AQA^\tra$ and $t$, $0\leq t < k$, is an integer such that $(k-t)\lambda_t > \sum_{i >
    t}\lambda_i \geq (k-t)\lambda_{t+1}$. At least one of $\sum_{i =
    1}^t{\lambda_i^{1/2}}$ and $\sqrt{k-t}\left(\sum_{i >
      t}{\lambda_i}\right)^{1/2}$ must be bounded from below by
  $\frac{1}{2}\|\Sigma\|_{(k)}$. Next we consider these two cases separately.

  Assume first that $\sum_{i = 1}^t{\lambda_i^{1/2}} \geq
  \frac{1}{2}\|\Sigma\|_{(k)}$. Let $\Pi$ be the orthogonal projection
  operator onto the eigenspace of $AQA^\tra$
  corresponding to $\lambda_1, \ldots, \lambda_t$. Then, because
  $\lambda_1 \geq \ldots \geq \lambda_t$ are the nonzero singular
  values of $\Pi AQ^{1/2}$, we have $\|\Pi
  AQ^{1/2}\|_{\tr}=\sum_{i = 1}^t{\lambda_i^{1/2}} \geq
  \frac{1}{2}\|\Sigma\|_{(k)}$. By Lemma~\ref{lm:bt-lb} applied to
  the matrices $M = \Pi A$ and $W = Q$, there exists a set $S
  \subseteq \univ$ of size at most $|S| \leq \rank~\Pi A = t< \eps n$,
  such that
  \begin{align}
    \specLB(\eps n, A) &\geq
    \sqrt{\frac{|S|}{|\quer|}}\ \lambda_{\min}(A_S) \notag\\
    &\geq  \sqrt{\frac{|S|}{|\quer|}}\ \lambda_{\min}(\Pi A_S) \geq \frac{c
      \|\Pi AQ^{1/2}\|_{\tr}}{(\log \eps n)\sqrt{|\quer|}} \geq
    \frac{c\|\Sigma\|_{(k)}}{2(\log \eps n) \sqrt{|\quer|}}\label{eq:main-lb-small-1}
  \end{align}
  for an absolute constant $c$.

  For the second case, assume that $\sqrt{k-t}\left(\sum_{i >
      t}{\lambda_i}\right)^{1/2} \geq
  \frac{1}{2}\|\Sigma\|_{(k)}$. Let $\Pi$ now be an orthogonal
  projection operator onto the eigenspace of $AQA^\tra$ corresponding
  to $\lambda_{t+1},\ldots, \lambda_{m}$. By the choice of $t$, we
  have
  \[
  \frac{\tr(\Pi AQA\Pi)}{\|\Pi AQA\Pi\|_2} =
  \frac{\sum_{i > t}\lambda_i}{\lambda_{t+1}} \geq k-t.
  \]
  By Theorem~\ref{thm:bt}, applied with $M = \Pi A$, $W = Q$, and
  $\eps = \frac{1}{2}$, there exists a set $S \subseteq U$ of size
  $\frac{1}{4}(k-t) < k \leq \eps n$ so that
  \begin{align}
    \specLB_2(\eps n, A) &\geq
    \sqrt{\frac{|S|}{|\quer|}}\lambda_{\min}(A_S) \notag\\
    &\geq  \sqrt{\frac{|S|}{|\quer|}}\lambda_{\min}(\Pi A_S) 
    \geq \frac{\sqrt{k-t}\left(\sum_{i > t}\lambda_i\right)^{1/2} }{4\sqrt{|\quer|}}
    \geq \frac{\|\Sigma\|_{(k)}}{8\sqrt{|\quer|}}.\label{eq:main-lb-small-2}
  \end{align}
  The theorem follows from \eqref{eq:proj-err}, the fact that at least
  one of \eqref{eq:main-lb-small-1} or \eqref{eq:main-lb-small-2} holds,
  and Theorem~\ref{thm:speclb-small}. 
\end{proof}

\section{Conclusion}

Several natural problems remain open. Probably the most important one
is to prove results analogous to ours for worst case, rather than
average, error. In that case the simple post-processing strategy of
the projection mechanism will likely not be sufficient. Another
interesting problem is to remove the dependence on the universe size
in the competetiveness ratio. It is plausible that this can be done
with the projection mechanism and a well-chosen Gaussian noise
distribution, but we would need tighter lower bounds, possibly based
on fingerprinting codes as in~\cite{BunUV13}.

\section*{Acknowledgments}
The author would like to thank the anonymous reviewers of ICALP 2015 for helpful comments.

%
%
\bibliographystyle{alpha}
\bibliography{Privacy,mypapers}

\end{document}